\documentclass[a4paper,USenglish,cleveref, autoref, numberwithinsect]{lipics-v2019}

\title{Accurate MapReduce Algorithms for $k$-median and $k$-means in General Metric Spaces} 

\titlerunning{Accurate MapReduce Algorithms for $k$-median/$k$-means in General Metric Spaces}

\author{Alessio Mazzetto\footnote{This work was done while the author was a graduate student at University of Padova}}{Department of Computer Science, Brown University, Providence, USA}{alessio\_mazzetto@brown.edu}{}{}
\author{Andrea Pietracaprina}{Department of Information Engineering, University of Padova, Padova, Italy}{andrea.pietracaprina@unipd.it}{}{}
\author{Geppino Pucci}{Department of Information Engineering, University of Padova, Padova, Italy}{geppino.pucci@unipd.it}{}{}

\authorrunning{A. Mazzetto, A. Pietracaprina and G. Pucci}

\Copyright{Alessio Mazzetto, Andrea Pietracaprina, and Geppino Pucci}



\funding{This work was supported, in part, by the University of Padova under grant
SID2017 and by MIUR, the Italian Ministry of Education,
University and Research, under grant PRIN \emph{AHeAD: efficient Algorithms for
HArnessing networked Data} and grant L.\ 232 \emph{``Dipartimenti di Eccellenza’’}}

\nolinenumbers 

\hideLIPIcs  

\EventEditors{Pinyan Lu and Guochuan Zhang}
\EventNoEds{2}
\EventLongTitle{30th International Symposium on Algorithms and Computation (ISAAC 2019)}
\EventShortTitle{ISAAC 2019}
\EventAcronym{ISAAC}
\EventYear{2019}
\EventDate{December 8--11, 2019}
\EventLocation{Shanghai University of Finance and Economics, Shanghai, China}
\EventLogo{}
\SeriesVolume{149}
\ArticleNo{37}

\newcommand{\inst}{\mathcal{I}}

\newcommand{\opt}{\texttt{\rm opt}}

\DeclareMathOperator*{\argmin}{arg\,min}
\usepackage[linesnumbered,ruled]{algorithm2e}

\begin{document}

\maketitle

\begin{abstract}
Center-based clustering is a fundamental primitive for data analysis
and becomes very challenging for large datasets. In this paper, we
focus on the popular $k$-median and $k$-means variants which, given a
set $P$ of points from a metric space and a parameter $k<|P|$, require
to identify a set $S$ of $k$ centers minimizing, respectively, the sum of the
distances and of the squared distances of all points in
$P$ from their closest centers. Our specific focus is on general
metric spaces, for which it is reasonable to require that the centers belong to the
input set (i.e., $S \subseteq P$). We present coreset-based 3-round
distributed approximation algorithms for the above problems using the
MapReduce computational model. The algorithms are rather simple and
obliviously adapt to the intrinsic complexity of the dataset, captured by the
doubling dimension $D$ of the metric space. Remarkably, the algorithms attain
approximation ratios that can be made arbitrarily close to those
achievable by the best known polynomial-time sequential approximations,
and they are very space efficient for small $D$, requiring local
memory sizes substantially sublinear in the input size. To the best
of our knowledge, no previous distributed approaches were able to attain
similar quality-performance guarantees in general metric spaces.
\end{abstract}

\section{Introduction}
\label{sec:introduction}
Clustering is a fundamental primitive in the realms of
data management and machine learning, with applications in a
large spectrum of domains such as database search, bioinformatics,
pattern recognition, networking, operations research, and many more
\cite{HennigMMR15}. A prominent clustering subspecies 
is \emph{center-based clustering} whose goal is to
partition a set of data items into $k$ groups, where $k$ is an input
parameter, according to a notion of similarity, captured by a given
measure of closeness to suitably chosen representatives, called
centers. There is a vast and well-established literature on sequential
strategies for different instantiations of center-based clustering
\cite{AwasthiB15}. However, the explosive growth of data that needs to
be processed often rules out the use of these sequential strategies,
which are often impractical on large data sets, due to
their time and space requirements. Therefore, it is of paramount
importance to devise efficient distributed clustering strategies
tailored to the typical computational frameworks for big data
processing, such as MapReduce \cite{LeskovecRU14}.

In this paper, we focus on the \emph{$k$-median} and \emph{$k$-means}
clustering problems.  Given a set $P$ of points in a general metric
space and a positive integer $k \leq |P|$, the $k$-median (resp.,
$k$-means) problem requires to find a subset $S \subseteq P$ of $k$
points, called \emph{centers}, so that the sum of all distances
(resp., square distances) between the points of $P$ to their closest
center is minimized. Once $S$ is determined, the association of
each point to its closest center naturally defines a clustering of
$P$. While scarcely meaningful for general metric spaces,
for Euclidean spaces, the widely studied \emph{continuous} variant of these two
problems removes the constraint that $S$ is a subset of $P$, hence
allowing a much richer choice of centers from the entire space.  Along
with \emph{$k$-center}, which requires to minimize the maximum
distance of a point to its closest center, $k$-median and $k$-means
are the most popular instantiations of center-based clustering, whose
efficient solution in the realm of big data has attracted vast
attention in the recent literature 
\cite{EneIM11,BahmaniMVKV12,BalcanEL13,Song0H17,CeccarelloPP19}.
One of the reference models for big data computing, also adopted in
most of the aforementioned works, is MapReduce 
\cite{DeanG08,PietracaprinaPRSU12,LeskovecRU14}, where a set of
processors with limited-size local memories process data in a sequence
of parallel rounds. Efficient MapReduce algorithms should aim at
minimizing the number of rounds while using substantially sublinear
local memory.

A natural approach to solving large instances of combinatorial
optimization problems relies on the extraction of a much smaller
``summary'' of the input instance, often dubbed \emph{coreset} in the
literature \cite{Har-Peled2004}, which embodies sufficient information
to enable the extraction of a good approximate solution of the whole
input.  This approach is profitable whenever the (time and space)
resources needed to compute the coreset are considerably lower than
those required to compute a solution by working directly on the input
instance. Coresets with different properties have been studied in the
literature to solve different variants of the aforementioned
clustering problems \cite{Philips2016}.

The main contributions of this paper are novel coreset-based
space/round-efficient MapReduce algorithms for $k$-median and
$k$-means.

\subsection{Related work}

The $k$-median and $k$-means clustering problems in general metric
spaces have been extensively studied, and constant approximation
algorithms are known for both problems \cite{AwasthiB15}.
In recent years, there has been growing interest in the
development of distributed algorithms to attack these problems in
the big data scenario (see \cite{Song0H17} and references therein).
While straightforward parallelizations of known iterative sequential
strategies tend to be inefficient due to high round complexity, the
most relevant efforts to date rely on distributed constructions of
coresets of size much smaller than the input, upon which a sequential
algorithm is then run to obtain the final solution.  
Ene et al.\ \cite{EneIM11} present a randomized MapReduce algorithm
which computes a coreset for $k$-median of size $O(k^2 |P|^{\delta})$
in $O(1/\delta)$ rounds, for any $\delta \in (0,1)$. By using an
$\alpha$-approximation algorithm on this coreset, a weak
$(10\alpha+3)$-approximate solution is obtained. In the paper, the
authors claim that their approach extends also to the $k$-means
problem, but do not provide the analysis.  For this latter problem,
in \cite{BahmaniMVKV12}  a parallelization of the popular
$k$-means++ algorithm by \cite{ArthurV07} is presented, which builds an 
$O(k \log |P|$)-size coreset for $k$-means
in $O(\log |P|)$ rounds. By running an $\alpha$-approximation
algorithm on the coreset, the returned solution features 
an $O(\alpha)$ approximation ratio. 
A randomized MapReduce algorithm for $k$-median has been
recently presented in
\cite{Song0H17}, where the well known local-search PAM algorithm
\cite{Kaufmann1987} is employed to extract a small family of
possible solutions from random samples of the input. A suitable
refinement of the best solution in the family is then
returned. While extensive experiments support the
effectiveness of this approach in practice, no tight theoretical
analysis of the resulting approximation quality is provided.

In the continuous setting, Balcan et al.\ \cite{BalcanEL13} present
randomized 2-round algorithms to build coresets in $\mathbb{R}^d$ of
size $O\left(\frac{kd}{\epsilon^2}+Lk\right)$ for $k$-median, and
$O\left(\frac{kd}{\epsilon^4 }+Lk\log({Lk})\right)$ for $k$-means,
for any choice of $\epsilon \in (0,1)$,
where the computation is distributed among $L$ processing elements.
By using an $\alpha$-approximation algorithm on the coresets, the
overall approximation factor is $\alpha+O(\epsilon)$. 
For $k$-means, a recent improved construction
yields a coreset which is a factor $O(\epsilon^2)$ smaller and features
very fast distributed implementation  \cite{BachemLK18}. It is not
difficult to show that a straightforward adaptation of these
algorithms to general spaces (hence in a non-continuous setting) would
yield $(c \cdot \alpha+ O(\epsilon))$-approximations, with $c \geq 2$,
thus introducing a non-negligible gap with respect to the quality of
the best sequential
approximations.

Finally, it is worth mentioning that there is a rich literature 
on sequential coreset constructions for $k$-median and $k$-means,
which mostly focus on the continuous case in Euclidean spaces
\cite{Feldman2011,Har-Peled2004,Har-Peled2005,SohlerW18,CohenCK18}.
We do not review the results in these works since our focus is on
distributed algorithms in general metric spaces. We also note that the recent work of \cite{Huang2018} 
addresses the construction of coresets for $k$-median and $k$-means in general metric spaces, where the 
coreset sizes are expressed as a function of the doubling dimension. However, their construction strategy is 
rather complex and it is not clear how to adapt it to the distributed setting. 

\subsection{Our contribution}

We devise new distributed coreset constructions and show how to employ
them to yield accurate space-efficient 3-round MapReduce algorithms for
$k$-median and $k$-means.  Our
coresets are built in a \emph{composable} fashion \cite{IndykMMM14} in
the sense that they are obtained as the union of small local coresets
computed in parallel (in 2 MapReduce rounds) on distinct subsets
of a partition of the input.  The final solution is obtained by
running a sequential approximation algorithm on the coreset in the
third MapReduce round. The memory requirements of our
algorithms are analyzed in terms of the desired approximation
guarantee, and of the \emph{doubling dimension} $D$ of the underlying
metric space, a parameter which generalizes the dimensionality of
Euclidean spaces to general metric spaces and is thus related to the
increasing difficulty of spotting good clusterings as the parameter
$D$ grows. 

Let $\alpha$ denote the best approximation ratio attainable
by a sequential algorithm for either $k$-median or $k$-means on general metric
spaces. Our main results are 
3-round $(\alpha+O(\epsilon))$-approximation 
MapReduce algorithms for $k$-median and $k$-means, which require
$O(|P|^{2/3}k^{1/3}\cdot$ \\ $(c/\epsilon)^{2D} \log^2{|P|})$ local memory, where
$c>0$ is a suitable constant that will be specified in the analysis,
and $\epsilon \in (0,1)$ is a user-defined precision parameter.
To the best of our knowledge, these are the first MapReduce algorithms
for $k$-median and $k$-means in general metric spaces which feature
approximation guarantees that can be made arbitrarily close to those
of the best sequential algorithms, and run in few rounds using local
space substantially sublinear for low-dimensional spaces.
In fact, prior to our work existing MapReduce algorithms for
$k$-median and $k$-means in general metric spaces
either exhibited approximation factors
much larger than $\alpha$ \cite{EneIM11,BahmaniMVKV12},
or missed a tight theoretical analysis of the approximation factor 
\cite{Song0H17}.

Our algorithms revolve around novel coreset constructions somehow
inspired by those proposed in \cite{Har-Peled2004} for Euclidean
spaces. As a fundamental tool, the constructions make use of a
procedure that, starting from a set of points $P$ and a set of centers
$C$, produces a (not much) larger set $C'$ such that for any point $x
\in P$ its distance from $C'$ is significantly smaller than its
distance from $C$.  Simpler versions of our constructions can also be
employed to attain 2-round MapReduce algorithms for the continuous
versions of the two problems, featuring $\alpha+O(\epsilon)$
approximation ratios. While similar approximation guarantees have
already been achieved in the literature using more space-efficient but randomized 
coreset constructions \cite{BalcanEL13,BachemLK18},  this result provides evidence of the general
applicability of our novel approach.

Finally, we want to point out that a very desirable feature of our
MapReduce algorithms is that they do not require a priori knowledge of
the doubling dimension $D$ and, in fact, it is easily shown that they
adapt to the dimensionality of the dataset which, in principle, can be
much lower than the one of the underlying space. 

\vspace*{0.3cm}
\noindent
{\bf Organization of the paper.}  The rest of the paper is organized
as follows.  Section~\ref{sec:preliminaries} contains a number of
preliminary concepts, including various properties of coresets that
are needed to achieve our results.
Section~\ref{section:coresetconstruct} presents our novel coreset
constructions for $k$-median
(Subsection~\ref{subsection:coresetkmedian}) and $k$-means
(Subsection~\ref{subsection:coresetkmeans}). Based on these
constructions, Subsection~\ref{subsection:mapreducefinal} derives the
MapReduce algorithms for the two problems. Finally,
Section~\ref{sec:conclusions} offers some concluding remarks.


\section{Preliminaries}
\label{sec:preliminaries}
Let $\mathcal{M}$ be a metric space with distance function
$d(\cdot,\cdot)$. We define the \emph{ball of radius $r$ centered at
  $x$} as the set of points at distance at most $r$ from
$x$. The \emph{doubling dimension} of $\mathcal{M}$ is the smallest
integer $D$ such that for any $r$ and $x \in \mathcal{M}$, the ball of
radius $r$ centered at $x$ can be covered by at most $2^D$ balls of
radius $r/2$ centered at points of $\mathcal{M}$.  Let $x \in
\mathcal{M}$ and $Y \subseteq \mathcal{M}$. We define $d(x,Y) =
\min_{y \in Y}d(x,y)$ and $x^Y = \argmin_{y \in Y}d(x,y)$.  A set of
points $P \subseteq \mathcal{M}$ can be weighted by assigning a
positive integer $w(p)$ to each $p \in P$.
In this case, we will use the notation $P_w$ (note that an 
unweighted set of points can be considered weighted with unitary weights).
Let $X_w$ and $Y$ be two subsets of $\mathcal{M}$. We define
$\nu_{X_w}(Y) = \sum_{x \in X_w}w(x)d(x,Y)$ and $\mu_{X_w}(Y) =
\sum_{x \in X_w}w(x)d(x,Y)^2$. The values $\nu_{X_w}(Y)$ and $\mu_{X_w}(Y)$
are also referred to as \textit{costs}.

In the \emph{$k$-median problem} (resp., \emph{$k$-means problem}), we
are given in input an instance $\inst = (P,k)$, with $P \subseteq
\mathcal{M}$ and $k$ a positive integer. A set $S \subseteq P$ is a
solution of $\inst$ if $|S| \leq k$. The objective is to find the
solution $S$ with minimum cost $\nu_{P}(S)$ (resp., $\mu_{P}(S)$).
Given an instance $\mathcal{I}$ of one of these two problems, we
denote with $\opt_\inst$ its optimal solution. Moreover, for $\alpha
\geq 1$, we say that $S$ is an \emph{$\alpha$-approximate solution}
for $\mathcal{I}$ if its cost is within a factor $\alpha$ from the
cost of $\opt_\inst$. In this case, the value $\alpha$ is also called
approximation factor. An \emph{$\alpha$-approximation algorithm}
computes an $\alpha$-approximate solution for any input instance. The
two problems are immediately generalized to the case of weighted
instances $(P_w,k)$. In fact, all known approximations algorithms can
be straightforwardly adapted to handle weighted instances keeping the
same approximation quality.

Observe that the squared distance does not satisfy the triangle
inequality. During the analysis, we will use the following weaker
bound.
\begin{proposition}
\label{proposition:squareddistance}
Let $x,y,z \in \mathcal{M}$. For every $c>0$
we have that $d(x,y)^2 \leq (1+1/c)d(x,z)^2 + (1+c)d(z,y)^2$.
\end{proposition}
\begin{proof}
Let $a,b$ be two real numbers. Since $(a/\sqrt{c}-b\cdot\sqrt{c})^2 \geq 0$, we obtain that $2ab \leq a^2/c + c \cdot b^2 $. Hence, $(a+b)^2 \leq (1+1/c)a^2 + (1+c)b^2$. The proof follows since 
$d(x,y)^2 \leq \left[ d(x,z)+d(z,y)\right]^2$ by triangle inequality.
\end{proof}

A coreset is a small (weighted) subset of the input which summarizes
the whole data. The concept of summarization can be captured with the
following definition, which is commonly adopted to describe coresets
for $k$-means and $k$-median (e.g., \cite{Har-Peled2004,Feldman2011,Huang2018}).
\begin{definition}
\label{definition:strong} 
A weighted set of points $C_w$ is an $\epsilon$-approximate coreset of an
instance $\inst = (P,k)$ of $k$-median $($resp., $k$-means$)$
if for any solution $S$ of $\inst$ it holds
that $|\nu_P(S) - \nu_{C_w}(S)| \leq \epsilon \cdot \nu_P(S)$
$($resp., $|\mu_P(S) - \mu_{C_w}(S)| \leq \epsilon \cdot \mu_P(S)$$)$.
\end{definition}
Informally, the cost of any solution is approximately the same if
computed from the $\epsilon$-approximate coreset rather than from the
full set of points. In the paper we will also make use of the
following different notion of coreset (already used in
\cite{Har-Peled2004,EneIM11}), which upper bounds the aggregate
``proximity'' of the input points from the coreset as a function of
the optimal cost.
\begin{definition}
\label{definition:bounded} Let $\inst = (P,k)$ be an instance
of $k$-median $($resp., $k$-means$)$. 
A set of points $C_w$ is an $\epsilon$-bounded coreset of $\inst$ if it exists a map $\tau: P \rightarrow C_w$ such that 
$\sum_{x \in P}d(x,\tau(x)) \leq \epsilon \cdot \nu_P(\opt_\inst)$ 
$($resp., $\sum_{x \in P}d(x,\tau(x))^2 \leq \epsilon \cdot \mu_P(\opt_\inst)$$)$ and for any $x \in C_w$, $w(x) = |\{ y \in P: \tau(y) = x\}|$. 
We say that $C_w$ is weighted according to $\tau$.
\end{definition}
The above two kind of coresets are related, as shown in the following
two lemmas.
\begin{lemma}
\label{lemma:boundedtostrongkmedian}
Let $C_w$ be an $\epsilon$-bounded coreset of a $k$-median instance $\inst=(P,k)$. Then $C_w$ is also a $\epsilon$-approximate coreset of $\inst$.
\end{lemma}
\begin{proof}
Let $\tau$ be the map of the definition of $\epsilon$-bounded coreset. Let $S$ be a solution of $\inst$. Using triangle inequality, we can easily see that $d(x,S) - d(x,\tau(x)) \leq d(\tau(x),S)$ and $d(\tau(x),S) \leq d(\tau(x),x) + d(x,S)$ for any $x \in P$. Summing over all points in $P$, we obtain that
\begin{align*}
    \nu_{P}(S) - \sum_{x \in P}d(x,\tau(x)) \leq \nu_{C_w}(S) \leq  \sum_{x \in P}d(x,\tau(x)) + \nu_{P}(S)
\end{align*}
To conclude the proof, we observe that $\sum_{x \in P}d(x,\tau(x)) \leq \epsilon \cdot \nu_P(\opt_\inst) \leq \epsilon \cdot \nu_P(S)$.
\end{proof}

\begin{lemma}
\label{lemma:boundedtostrongkmeans}
Let $C_w$ be an $\epsilon$-bounded coreset of a $k$-means instance $\inst=(P,k)$. Then $C_w$ is also a $(\epsilon+2\sqrt{\epsilon})$-approximate coreset of $\inst$.
\end{lemma}
\begin{proof}
Let $\tau$ be the map of the definition of $\epsilon$-bounded coreset. Let $S$ be a solution of $\inst$. We want to bound the quantity $| \mu_{P}(S) - \mu_{C_w}(S) | = \sum_{x \in P} | d(x, S)^2 - d(\tau(x),S)^2 |$.
We rewrite $|d(x, S)^2 - d(\tau(x),S)^2|$ as $\left[ d(x,S) + d(\tau(x),S) \right] \cdot | d(x,S) - d(\tau(x),S)|$. 
By triangle inequality, we have that $d(x,S) \leq d(x,\tau(x)) + d(\tau(x),S)$ and $d(\tau(x),S) \leq d(\tau(x),x) + d(x,S)$. By combining these two inequalities, it results that $|d(x,S) - d(\tau(x), S)| \leq d(x,\tau(x))$. Moreover, $d(x,S) + d(\tau(x),S) \leq 2d(x,S) + d(x,\tau(x))$. Hence
\begin{eqnarray*}
| \mu_{P}(S) - \mu_{C_w}(S) |  & \leq & 
\sum_{x \in P} d(x, \tau(x))\left[ 2d(x,S) + d(x,\tau(x))\right] \\
& \leq & \epsilon \cdot \mu_{P}(S) + 2\sum_{x \in P}d(x,\tau(x))d(x,S)
\end{eqnarray*}
where we used the fact that $\sum_{x \in P}d(x,\tau(x))^2 \leq \epsilon \cdot \mu_{P}(\opt_\inst) \leq \epsilon \cdot \mu_{P}(S)$. We now want to bound the sum over the products of the two distances. Arguing as in the proof of \autoref{proposition:squareddistance}, we can write:
\begin{align*}
    2\sum_{x \in P}d(x,\tau(x))d(x,S) \leq  \sqrt{\epsilon} \cdot \sum_{x \in P}d(x,S)^2 + \frac{1}{\sqrt{\epsilon}} \sum_{x \in P}d(x,\tau(x))^2 \leq 2 \sqrt{\epsilon} \cdot \mu_{P}(S)
\end{align*}
To wrap it up, it results that $ | \mu_{P}(S) - \mu_{C_w}(S) | \leq (\epsilon + 2 \sqrt{\epsilon})\cdot \mu_{P}(S)$.
\end{proof}

In our work, we will build coresets by working in parallel over a partition of the input instance. The next lemma provides known results on the relations between the optimal solution of the whole input points and the optimal solution of a subset of the input points.
\begin{lemma}
\label{lemma:optimalsolrelation}
Let $C_w \subseteq P$. Let $\inst = (P,k)$ and $\inst' = (C_w,k)$. Then:
$($a$)$ $\nu_{C_w}(\opt_{\inst'}) \leq 2\nu_{C_w}(\opt_\inst)$;
and
$($b$)$ $\mu_{C_w}(\opt_{\inst'}) \leq 4\mu_{C_w}(\opt_\inst)$.
\end{lemma}
\begin{proof}
  We first prove point $(b)$. Let $X = \{ x^{C_w} : x \in \opt_\inst \}$. The set $X$ is a solution of $\inst'$. By optimality of $\opt_{\inst'}$, we have that $\mu_{C_w}(\opt_{\inst'}) \leq \mu_{C_w}(X)$. Also, by triangle inequality, it holds that  $\mu_{C_w}(X) \leq \sum_{x \in C_w}w(x)\left[ d(x,\opt_\inst) + d(x^{\opt_\inst}, X)\right]^2$. We observe that $d(x^{\opt_\inst}, X) \leq d(x, \opt_\inst)$ by definition of $X$. Thus, we obtain that $\mu_{C_w}(\opt_{\inst'}) \leq 4\mu_{C_w}(\opt_\inst)$. The proof of $(a)$ follows the same lines with a factor $2$ less since we do not square.
  \end{proof}

Bounded coresets have the
nice property to be \emph{composable}. That is, we can partition the
input points into different subsets and compute a bounded coreset
separately in each subset: the union of those coresets is a bounded
coreset of the input instance. This property, which is formally stated in the
following lemma, is crucial to develop efficient MapReduce algorithms
for the clustering problems.

\begin{lemma}
\label{lemma:unionbounded}
Let $\inst = (P,k)$ be an instance of
$k$-median $($resp., $k$-means$)$. Let $P_1,\ldots,P_L$ be a partition of $P$. For $\ell=1,\ldots,L$, let $C_{w,\ell}$ be an $\epsilon$-bounded coreset of $\inst_\ell = (P_\ell,k)$. Then $C_w = \cup_\ell C_{w,\ell}$ is a $2\epsilon$-bounded coreset $($resp., a $4\epsilon$-bounded coreset$)$ of $\inst$.
\end{lemma}

\begin{proof}
We prove the lemma for $k$-median. The proof for $k$-means is similar.
For $\ell=1,\ldots,L$, let $\tau_\ell$ be the map from $P_\ell$ to $C_{w,\ell}$ of \autoref{definition:bounded}. Now, for any $x \in P$, let $\ell$ be the integer such that $x \in P_\ell$; we define $\tau(x) = \tau_\ell(x)$.
\begin{align*}
    \sum_{x \in P}d(x,\tau(x)) \leq \sum_{\ell=1}^{L}\sum_{x \in P_\ell}d(x,\tau_\ell(x)) \leq \epsilon \sum_{\ell=1}^L \nu_{P_\ell}(\opt_{\inst_\ell}) \leq 2\epsilon \cdot \nu_{P}(\opt_{\inst})
\end{align*}
In the last inequality, we used the fact that $\nu_{P_\ell}(\opt_{\inst_\ell}) \leq  2 \nu_{P_\ell}(\opt_{\inst})$ from \autoref{lemma:optimalsolrelation}.
\end{proof}

In the paper, we will need the following additional characterization
of a representative subset of the input, originally introduced in
\cite{Har-Peled2004}.
\begin{definition}
\label{definition:centroidset}
Let $\inst = (P,k)$ be an instance of $k$-median $($resp.,
$k$-means$)$.  A set $C$ is said to be an $\epsilon$-centroid set of
$\inst$ if there exists a subset $X \subseteq C$, $|X| \leq k$, such
that $\nu_P(X) \leq (1+\epsilon)\nu_P(\opt_\inst)$
$($resp., $\mu_P(X) \leq (1+\epsilon)\mu_P(\opt_\inst)$$)$.
\end{definition}
\noindent 
Our algorithms are designed for the \emph{MapReduce}
model of computation which has become a
de facto standard for big data algorithmics  in recent years.
A MapReduce
algorithm~\cite{DeanG08,PietracaprinaPRSU12,LeskovecRU14} executes in
a sequence of parallel \emph{rounds}. In a round, a multiset $X$ of
key-value pairs is first transformed into a new multiset $X'$ of
key-value pairs by applying a given \emph{map function} (simply called
\emph{mapper}) to each individual pair, and then into a final multiset
$Y$ of pairs by applying a given \emph{reduce function} (simply called
\emph{reducer}) independently to each subset of pairs of $X'$ having
the same key.  The model features two parameters, $M_L$, the
\emph{local memory} available to each mapper/reducer, and $M_A$, the
\emph{aggregate memory} across all mappers/reducers. 

\section{Coresets construction in MapReduce}
\label{section:coresetconstruct}
\sloppy
Our coreset constructions are based on a suitable point selection
algorithm called $\texttt{CoverWithBalls}$, somewhat inspired by the exponential
grid construction used in \cite{Har-Peled2004} to build
$\epsilon$-approximate coresets in $\mathbb{R}^d$ for the continuous
case. Suppose that we want to build an $\epsilon$-bounded coreset of a
$k$-median instance $\inst = (P,k)$ and that a $\beta$-approximate
solution $T$ for $\inst$ is available. A simple approach would be to
find a set $C_w$ such that for any $x$ in $P$ there exists a point
$\tau(x) \in C$ for which $d(x,\tau(x)) \leq (\epsilon/2\beta) \cdot
d(x,T)$. Indeed, if $C_w$ is weighted according to $\tau$, it can be seen
that $C_w$ is an $\epsilon$-bounded coreset of
$\inst$. The set $C_w$ can be constructed greedily by iteratively
selecting an arbitrary point $p \in P$, adding it to $C_w$, and
discarding all points $q \in P$ (including $p$) 
for which the aforementioned property
holds with $\tau(q) = p$. The construction ends when all points of $P$
are discarded. However, note that the points of $P$ which are already very
close to $T$, say at a distance $\leq R$ for a suitable tolerance
threshold $R$, do not contribute much to $\nu_{P}(T)$, and so to the sum
$\sum_{x \in P}d(x,\tau(x))$. For these points, we can relax the
constraint and discard them from $P$ as soon their distance to
$C_w$ becomes at most $(\epsilon/2\beta) \cdot R$. This relaxation is
crucial to bound the size of the returned set as a function of the
doubling dimension of the space. 
\begin{algorithm}
    \DontPrintSemicolon
    $C_w \leftarrow \emptyset$ \;
    \While{$P \neq \emptyset$}{
        $p \longleftarrow $ arbitrarily selected point in $P$ \;
        $C_w \longleftarrow C_w \cup \{ p \}, w(p) \longleftarrow 0$ \;
        \ForEach{$q \in P$}{
            \If{ $d(p,q) \leq \epsilon/(2\beta) \max \{R,d(q,T) \} $}{
                remove $q$ from $P$ \;
                $w(p) \longleftarrow w(p)+1$ \hspace{10pt} \tcc*[r]{(i.e. $\tau(q) = p $, see \autoref{lemma:taucoverwithballs})}
            }         
        }
    }
    return $C_w$
    \caption{\texttt{CoverWithBalls}$(P,T,R,\epsilon,\beta)$}
\end{algorithm}
Algorithm $\texttt{CoverWithBalls}$ 
is formally described in the
pseudocode below. It receives in input 
two sets of points, $P$ and $T$, and three positive
real parameters $R$, $\epsilon$, and $\beta$, with $\epsilon < 1$ and
$\beta \geq 1$  and  outputs a weighted set $C_w \subseteq
P$ which satisfies the property stated in the following lemma.
\begin{lemma}
\label{lemma:taucoverwithballs}
Let $C_w$ be the output of $\texttt{CoverWithBalls}(P,T,R,\epsilon,\beta)$. 
$C_w$ is weighted according to  a map $\tau: P \rightarrow C_w$ such that, for any $x \in P$, $d(x,\tau(x)) \leq \epsilon/(2\beta)\max\{R,d(x,T)\}$.
\end{lemma}
\begin{proof}
For any $x \in P$, we define $\tau(x)$ as the point in $C_w$ which caused the removal of $x$ from $P$ during the execution of the algorithm. The statement immediately follows.
\end{proof}
While in principle the size of $C_w$ can be arbitrarily close to
$|P|$, the next theorem shows that this is not the case for low
dimensional spaces, as a consequence of the fact that there cannot be
too many points which are all far from one another.  We first
need a technical lemma. A set of points $X$ is said to be an
\emph{$r$-clique} if for any $x,y \in X$, $x \neq y$, it holds that
$d(x,y) > r$. We have:
\begin{lemma}
\label{lemma:sizeclique}
Let $0 < \epsilon < 1$. Let $\mathcal{M}$ be a metric space with doubling dimension $D$. Let $X \subseteq \mathcal{M}$ be an $\epsilon \cdot r$-clique and assume that $X$ can be covered by a ball of radius $r$ centered at a point of $\mathcal{M}$. Then, $|X| \leq (4/\epsilon)^D$. 
\end{lemma}
\begin{proof}
By recursively applying the definition of doubling dimension, we observe that the ball of radius $r$ which covers $X$ can be covered by $2^{j\cdot D}$ balls of radius $2^{-j}\cdot r$, where $j$ is any non negative integer. Let $i$ be the least integer for which $2^{-i}\cdot r \leq \epsilon/2 \cdot r$ holds. Any of the $2^{i \cdot D}$ balls with radius $2^{-i}\cdot r$ can contain at most one point of $X$, since $X$ is a $\epsilon\cdot r$-clique. Thus $|X| \leq 2^{i \cdot D}$.
As $i = 1 + \lceil \log_2{(1/\epsilon)} \rceil$, we finally obtain that $|X| \leq (4/\epsilon)^D$.
\end{proof}
\begin{theorem}
\label{theorem:sizecoverwithballs}
Let $C_w$ be the set returned by the execution of
$\texttt{CoverWithBalls}(P,T,R,\epsilon,\beta)$. Suppose that the
points in $P$ and $T$ belong to a metric space with doubling dimension
$D$. Let $c$ be a real value such that, for any $x \in P$, $c \cdot R
\geq d(x,T)$. Then,
\begin{align*}
    |C_w| \leq |T| \cdot \left(16\beta/\epsilon\right)^D \cdot (\log_2{c} + 2)
\end{align*}
\end{theorem}
\begin{proof}
Let $T = \{t_1,\ldots,t_{|T|} \}$ be the set in input to the algorithm. For any $i$, $1 \leq i \leq |T|$, let $P_i = \{ x \in P: x^T = t_i \}$
and  $B_{i} = \{x \in P_i: d(x,T) \leq R \}$. In addition, for any integer value $j \geq 0$ and for any feasible value of $i$, we define $D_{i,j} = \{ x \in P_i: 2^{j}\cdot R < d(x,T) \leq 2^{j+1}\cdot R \}$. We observe that for any $j \geq \lceil \log_2{c} \rceil$, the sets $D_{i,j}$ are empty, since $d(x,T) \leq c \cdot R$. Together, the sets $B_{i}$ and $D_{i,j}$ are a partition of $P_i$. 

For any $i$, let $C_{i} = C_w \cap B_{i}$. We now want to show that the set $C_{i}$ is a $\epsilon/(2\beta)\cdot R$-clique. Let $c_1,c_2$ be any two different points in $C_{i}$ and suppose, without loss of generality, that $c_1$ was added first to $C_w$. Since $c_2$ was not removed from $P$, this means that $d(c_1,c_2) > \epsilon/(2\beta)\cdot \max \{ d(c_2,T), R \} \geq \epsilon/(2\beta)R$, where we used the fact that $d(c_2,T) \leq R$ since $c_2$ belongs to $B_i$. Also, the set $C_i \subseteq B_{i}$ is contained in a ball of radius $R$ centered in $t_i$, thus we can apply \autoref{lemma:sizeclique} and bound its size, obtaining that $|C_i| \leq (8\beta/\epsilon)^D$.

For any $i$ and $j$, let $C_{i,j} = C_w \cap D_{i,j}$. We can use a
similar strategy to bound the size of those sets. We first show that
the sets $C_{i,j}$ are $\frac{\epsilon}{4\beta}\cdot
2^{j+1}R$-cliques. Let $c_1,c_2$ be any two different points in
$C_{i,j}$ and suppose, without loss of generality, that $c_1$ was
added first to $C_w$. Since $c_2$ was not removed from $P$, this means
that $d(c_1,c_2) > \epsilon/(2\beta)\cdot \max \{ d(c_2,T), R \} \geq
\epsilon/(4\beta)2^{j+1}R$, where we used the fact that $d(c_2,T) >
2^j \cdot R$ since $c_2$ belongs to $D_{i,j}$. Also, the set $C_{i,j}
\subseteq D_{i,j}$ is contained in a ball of radius $2^{j+1}R$
centered in $t_i$, thus we can apply \autoref{lemma:sizeclique} and
obtain that $|C_{i,j}| \leq (16\beta/\epsilon)^D$. Since the sets
$C_i$ and $C_{i,j}$ partition $C_w$, we can bound the size of $C_w$ as
the sum of the bounds of the size of those sets. Hence:
\begin{align*}
    |C_w| \leq \sum_{i=1}^{|T|} |C_{i}| + \sum_{i=1}^{|T|} \sum_{j=0}^{\lceil \log_2{c} \rceil - 1}|C_{i,j}| \leq |T|\cdot (16\beta/\epsilon)^D \cdot (\log_2{c}+2)
\end{align*}
\end{proof}


\subsection{A first approach to coreset construction for $k$-median}
\label{subsection:approachkmedian}
In this subsection we present a $1$-round MapReduce algorithm that
builds a weighted coreset $C_w \subseteq P$ of a $k$-median instance $\inst =
(P,k)$. The algorithm is parametrized by a value $\epsilon \in (0,1)$, 
which represents a tradeoff between coreset size and
accuracy. The returned coreset has the following property. Let $\inst'
= (C_w,k)$. If we run an $\alpha$-approximation algorithm on 
$\inst'$, then the
returned solution is a $(2\alpha+O(\epsilon))$-approximate
solution of $\inst$. 
Building on this construction, in the next subsection
we will obtain a better coreset which allows us to reduce the
final approximation factor to the desired $\alpha+O(\epsilon)$ value.
The coreset construction algorithm operates as follows. The
set $P$ is partitioned into $L$ equally-sized subsets
$P_1,\ldots,P_L$. In parallel, on each $k$-median instance
$\inst_\ell = (P_\ell,k)$, with $\ell=1,\ldots,L$, the
following operations are performed:
\begin{enumerate}
\item Compute a set $T_\ell$ of $m \geq k$ points 
such that $\nu_{P_\ell}(T_\ell) \leq \beta \cdot \nu_{P_\ell}(\opt_{\inst_\ell})$.
\item $R_\ell \longleftarrow  \nu_{P_\ell}(T_\ell)/|P_\ell|$.
\item $C_{w,\ell} \longleftarrow \texttt{CoverWithBalls}(P_\ell,T_\ell,R_\ell,\epsilon,\beta)$.
\end{enumerate}
The set $C_w = \cup_{\ell=1}^{L}C_{w,\ell}$ is the output of the
algorithm. 

In Step 1, the set $T_\ell$ can be computed through a sequential
(possibly bi-criteria) approximation algorithm for $m$-median, with a
suitable $m \geq k$, to yield a small value of $\beta$.  If we assume
that such an algorithm requires space linear in $P_\ell$, the entire
coreset costruction can be implemented in a single MapReduce round,
using $O(|P|/L)$ local memory and $O(|P|)$ aggregate memory. For
example, using one of the known linear-space, constant-approximation
algorithms (e.g., \cite{AryaGKMMP04}), we can get $\beta = O(1)$ with
$m=k$.

\begin{lemma}
\label{lemma:cwboundedcoreset}
For $\ell=1,\ldots,L$, $C_{w,\ell}$ is an $\epsilon$-bounded coreset 
of the $k$-median instance $\inst_\ell$.
\end{lemma}
\begin{proof}
Fix a value of $\ell$. Let $\tau_\ell$ be the map between the points in $C_{w,\ell}$ and the points in $P_\ell$ of \autoref{lemma:taucoverwithballs}. The set $C_{w,\ell}$ is weighted according to $\tau_\ell$. Also, it holds that:
\begin{align*}
    \sum_{x \in P_\ell}d(x,\tau_\ell(x)) \leq \frac{\epsilon}{2\beta} \sum_{x \in P_\ell}\left(R_\ell+d(x,T_\ell) \right) \leq  \frac{\epsilon}{2\beta}\left( R_\ell\cdot|P_\ell|+\nu_{P_\ell}(T_\ell) \right) \leq \epsilon \cdot \nu_{P_\ell}(\opt_{\inst_\ell})
\end{align*}
\end{proof}
By combining \autoref{lemma:cwboundedcoreset} and \autoref{lemma:unionbounded}, the next lemma immediately follows.
\begin{lemma}
\label{lemma:coresetkmedian}
Let $\inst = (P,k)$ be a $k$-median instance. The set $C_w$ returned by the above MapReduce algorithm is a $2\epsilon$-bounded coreset of $\inst$.
\end{lemma}

It is possible to bound the size of $C_w$ as a function of the doubling
dimension $D$. For any
$\ell=1,\ldots,L$ and $x \in P_\ell$, it holds that $R_{\ell}\cdot
|P_\ell| = \nu_{P_\ell}(T_\ell) \geq d(x,T_\ell)$, thus we can
bound the size of $C_{w,\ell}$ by using
\autoref{theorem:sizecoverwithballs}. 
Since $C_w$ is the union of
those sets, this argument proves the following lemma.

\begin{lemma}
Let $\inst = (P,k)$ be a $k$-median instance. Suppose that the points in $P$ belong to a metric space with doubling dimension $D$. Let $C_w$ be the set
 returned by the above MapReduce algorithm with input $\inst$ and $m\geq k$. Then,  $|C_w| = O\left( L\cdot m \cdot (16\beta/\epsilon)^{D} \log{|P|}  \right)$
\end{lemma}

Let $S$ be an $\alpha$-approximate solution of $\inst' =(C_w,k)$, with constant $\alpha$. We will now show that $\nu_{P}(S)/\nu_{P}(\opt_\inst) = 2\alpha + O(\epsilon)$. 
Let $\tau$ be the map of 
from $P$ to $C_w$ (see \autoref{lemma:taucoverwithballs}). By triangle inequality, $\nu_{P}(S) \leq \sum_{x \in P}d(x,\tau(x)) + \nu_{C_w}(S)$. We have that $\sum_{x \in P}d(x,\tau(x)) \leq 2\epsilon \cdot \nu_{P}(\opt_\inst)$ since,
by \autoref{lemma:coresetkmedian},
 $C_w$ is a $2\epsilon$-bounded coreset. 
By the fact that $S$ is an $\alpha$-approximate solution of $\inst'$ and by
\autoref{lemma:optimalsolrelation}, we have that 
$\nu_{C_w}(S) \leq \alpha \cdot \nu_{C_w}(\opt_{\inst'}) \leq 2\alpha \cdot \nu_{C_w}(\opt_\inst)$. By \autoref{lemma:boundedtostrongkmedian}, $C_w$ is also a $2\epsilon$-approximate coreset of $\inst$, thus $\nu_{C_w}(\opt_\inst) \leq (1+2\epsilon)\nu_{P}(\opt_\inst)$. Putting it all together, we have that $\nu_{P}(S)/\nu_{P}(\opt_\inst) \leq 2\alpha(1+2\epsilon)+2\epsilon = 2\alpha + O(\epsilon)$.
We observe that the factor $2$ is due to the inequality which relates $\opt_\inst$ and $\opt_{\inst'}$, namely $\nu_{C_w}(\opt_{\inst'}) \leq 2\nu_{C_w}(\opt_\inst)$. In the next subsection, we will show how to get rid of this 
factor.

\paragraph*{Application to the continuous case}
The same algorithm of this subsection can also be used to build a
$O(\epsilon)$-approximate coreset in the continuous scenario where
centers are not required to belong to $P$. It is easy to verify that
the construction presented in this subsection also works in the
continuous case, with the final 
approximation factor improving to $(\alpha +
  O(\epsilon))$.
Indeed, we can use the stronger inequality
$\nu_{C_w}(\opt_{\inst'}) \leq \nu_{C_w}(\opt_\inst)$, as
$\opt_{\inst}$ is also a solution of $\inst'$, which allows us to avoid
the factor $2$ in front of $\alpha$. While
the same approximation guarantee has
already been achieved in the literature using more space-efficient but
randomized coreset constructions 
\cite{BalcanEL13,BachemLK18}, as mentioned in the
introduction, this result provides evidence of the general
applicability of our approach.


\subsection{Coreset construction for $k$-median}
\label{subsection:coresetkmedian}
In this subsection, we present a $2$-round MapReduce algorithm which
computes a weighted subset which is both an $O(\epsilon)$-bounded
coreset and an $O(\epsilon)$-centroid set of an input instance $\inst
= (P,k)$ of $k$-median. The algorithm is similar to the one of the
previous subsection, but applies $\texttt{CoverWithBalls}$ twice in
every subset of the partition. This idea is inspired by the strategy
presented in \cite{Har-Peled2004} for $\mathbb{R}^d$, where a double
exponential grid construction is used to ensure that the returned
subset is a centroid set. 

\noindent{\bf First Round.}
$P$ is partitioned into $L$ equally-sized subsets $P_1,\ldots,P_L$. Then in
parallel, on each $k$-median instance $\inst_\ell = (P_\ell,k)$, with
$\ell=1,\ldots,L$, the following steps are performed:
\begin{enumerate}
\item Compute a set $T_\ell$ of $m\geq k$ points such that $\nu_{P_\ell}(T_\ell) \leq \beta \cdot \nu_{P_\ell}(\opt_{\inst_\ell})$.
\item $R_\ell \longleftarrow \nu_{P_\ell}(T_\ell)/|P_\ell|$.
\item $C_{w,\ell} \longleftarrow \texttt{CoverWithBalls}(P_\ell,T_\ell,R_\ell,\epsilon,\beta)$.
\end{enumerate}
\noindent{\bf Second Round.} 
Let $C_w = \cup_{\ell =1}^{L}C_{w,\ell}$. The same partition of $P$ of the first round is used.  Together with $P_{\ell}$, the $\ell$-th reducer  receives a copy of  $C_w$, and all values $R_i$ computed in the previous round, for $i = 1, \ldots, L$. On each $k$-median instance $\inst_\ell = (P_\ell,k)$, with $\ell = 1,\ldots,L$, the following steps are performed:
\begin{enumerate} 
\item $R \longleftarrow \sum_{i=1}^{L} |P_i| \cdot R_i / |P|$
\item $E_{w,\ell} \longleftarrow \texttt{CoverWithBalls}(P_\ell,C_{w},R,\epsilon,\beta)$.
\end{enumerate}

The set $E_w = \cup_{\ell=1}^{L}E_{w,\ell}$ is the output of the
algorithm.  The computation of $T_{\ell}$ in the first round is accomplished as described in the previous section. 

The following lemma characterizes the properties of $E_w$. 
\begin{lemma}
\label{lemma:centroidkmedian}
Let $\inst = (P,k)$ be a $k$-median instance. Then, the set  $E_w$ returned by
the above MapReduce algorithm is both a $2\epsilon$-bounded coreset
and a $7\epsilon$-centroid set of $\inst$.
\end{lemma}
\begin{proof}
The first three steps of the algorithm are in common with the
algorithm of \autoref{subsection:coresetkmedian}. By
\autoref{lemma:cwboundedcoreset}, for $\ell=1,...,L$, the sets
$C_{w,\ell}$ are $\epsilon$-bounded coresets of $\inst_\ell$. Let $C_w
= \cup_{\ell=1}^{L}C_{w,\ell}$. By \autoref{lemma:unionbounded}, the
set $C_w$ is a $2\epsilon$-bounded coreset of $\inst$, and also, by
\autoref{lemma:boundedtostrongkmedian}, a $2\epsilon$-approximate
coreset. Let $\tau(x)$ be the map from $P$ to $C_w$ as specified
in \autoref{definition:bounded}. It holds that $\nu_{P}(C_w) \leq
\sum_{x \in P}d(x,\tau(x)) \leq 2\epsilon \cdot \nu_{P}(\opt_\inst)$. Let $\phi_\ell$ be the map of
\autoref{lemma:taucoverwithballs} from the points in $P_\ell$ to the
points in $E_{w,\ell}$. By reasoning as in the proof of
\autoref{lemma:cwboundedcoreset}, we obtain that $\sum_{x \in
P_\ell}d(x, \phi_\ell(x)) \leq \epsilon/(2\beta)\left[|P_\ell| \cdot
R + \nu_{P_\ell}(C_{w})\right]$. For any $x \in P$, let $\hat{\ell}$ be the index for which $x \in P_{\hat{\ell}}$, we define $\phi(x) = \phi_{\hat{\ell}}(x)$. We have that
\begin{align*}
  \sum_{x \in P}d(x,\phi(x)) \leq \frac{\epsilon}{2\beta}\sum_{\ell = 1}^{L} \left[ R\cdot |P_\ell| + \nu_{P_\ell}(C_w) \right] = \frac{\epsilon}{2\beta}\left(\left(\sum_{\ell=1}^{L} | P_\ell | \cdot R_\ell\right) + \nu_{P}(C_w)  \right)
\end{align*}
where in the last equality we applied the definition of $R$. Since
$|P_\ell|\cdot R_\ell = \nu_{P_\ell}(T_\ell) \leq \beta \cdot
\nu_{P_\ell}(\opt_{\inst_\ell}) \leq 2\beta\cdot\nu_{P_\ell}(\opt_\inst)$,  where the last inequality follows from \autoref{lemma:optimalsolrelation}, we have that $\sum_{\ell=1}^{L}
  |P_\ell|\cdot R_\ell \leq 2\beta \cdot \nu_{P}(\opt_\inst)$. Additionally, $\nu_{P}(C_w) \leq 2\epsilon \cdot \nu_{P}(\opt_\inst)$ as argued previously in the proof. Therefore $E_w$ is a $2\epsilon$-bounded coreset. 

We now show that $E_w$ is a $7\epsilon$-centroid set of $\inst$. Let $X = \{ x^{E_w} : x \in \opt_\inst \}$. We will prove that $\nu_{P}(X) \leq (1+7\epsilon)\nu_{P}(\opt_\inst)$. By triangle inequality, we obtain that:
\begin{align*}
    \nu_{P}(X) = \sum_{x \in P}d(x,X) \leq \sum_{x \in P}d(x,\tau(x))+\sum_{x \in P}d(\tau(x),X)
\end{align*}
The first term of the above sum can be bounded as $\sum_{x \in P}d(x,\tau(x)) \leq 2\epsilon \cdot \nu_{P}(\opt_\inst)$, since $C_w$ is a $2\epsilon$-bounded coreset. Also, we notice that the second term of the sum can be rewritten as $\sum_{x \in P}d(\tau(x),X) = \sum_{x \in C_w}w(x)d(x,X)$, due to the relation between $\tau$ and $w$. By triangle inequality, we obtain that:
\begin{align*}
    \sum_{x \in C_w}w(x)d(x,X) \leq \sum_{x \in C_w}w(x)d(x,x^{\opt_\inst})+\sum_{x \in C_w}w(x)d(x^{\opt_\inst},X)
\end{align*}
Since $C_w$ is a $2\epsilon$-approximate coreset, we can use the bound
$\sum_{x \in C_w}w(x)d(x,x^{\opt_\inst}) = \nu_{C_w}(\opt_{\inst})
\leq (1+2\epsilon)\nu_{P}(\opt_\inst)$.  Also, by using the definition of $X$, we observe that
\begin{align*}
    \sum_{x \in C_w}w(x)d(& x^{\opt_\inst}, X) = \sum_{x \in C_w}w(x)d(x^{\opt_\inst},E_w) \leq \sum_{x \in C_w}w(x)d(x^{\opt_\inst},\phi(x^{\opt_\inst}))\\
    &\leq \frac{\epsilon}{2\beta} \sum_{x \in C_w} w(x) \cdot \left( R + 
d(x^{\opt_\inst}, C_w) \right) \leq \frac{\epsilon}{2\beta}\left(\left( \sum_{\ell=1}^{L} |P_\ell| \cdot R_\ell\right) + \nu_{C_w}(\opt_\inst) \right)
\end{align*}
In the last inequality, we used the definition of $R$, and the simple observation that for any $x \in C_{w}$, $d(x^{\opt_\inst},C_{w}) \leq d(x,x^{\opt_\inst}) = d(x,\opt_\inst)$. As argued previously in the proof, we have that $\sum_{\ell} |P_\ell| \cdot R_\ell \leq 2\beta \cdot \nu_{P}(\opt_\inst)$. Also, $\nu_{C_w}(\opt_\inst) \leq (1+2\epsilon)\nu_{P}(\opt_\inst)$ as $C_w$ is a $2\epsilon$-approximate coreset of $\inst$. Since we assume that $\beta \geq 1$, we finally obtain:
\begin{align*}
\sum_{x \in C_w}w(x)d(x^{\opt_\inst},X) \leq \frac{\epsilon}{2\beta}( 2\beta + 1 + 2\epsilon) \nu_{P}(\opt_\inst) \leq 3\epsilon \cdot \nu_{P}(\opt_\inst)
\end{align*}
We conclude that $\nu_{P}(X) \leq (2\epsilon+1+2\epsilon+3\epsilon)\nu_{P}(\opt_\inst) = (1+7\epsilon) \cdot \nu_{P}(\opt_\inst)$
\end{proof}

The next lemma establishes an upper bound on the size of $E_w$.

\begin{lemma} \label{lemma:kmediansize}   
Let $\inst = (P,k)$ be a $k$-median instance. Suppose that the points in $P$ belong to a metric space with doubling dimension $D$. Let $E_w$ be the set returned by the above MapReduce algorithm with input $\inst$ and $m\geq k$. Then $|E_w| = O\left( L^2\cdot m \cdot (16\beta/\epsilon)^{2D} \log^2{|P|}  \right)$.
\end{lemma}
\begin{proof}
From the previous subsection, we know that $|C_w| = O\left(L \cdot 
m\cdot(16\beta/\epsilon)^{D}\log{|P|} \right)$. Also, by \autoref{lemma:cwboundedcoreset}, we have that $\nu_{P_\ell}(C_{w,\ell}) \leq \epsilon \cdot \nu_{P_\ell}(\opt_{\inst_\ell})$ for any $\ell = 1, \ldots, L$. For every $x \in P$
we have that $\epsilon| P|\cdot R = \epsilon \sum_{\ell}|P_\ell|\cdot R_\ell = \epsilon\sum_{\ell}\nu_{P_\ell}(T_\ell) \geq  \sum_{\ell}\epsilon\cdot \nu_{P_\ell}(\opt_{\inst_\ell}) \geq
\sum_{\ell}\nu_{P_\ell}(C_{w,\ell}) \geq \nu_P(C_w) \geq d(x,C_{w}) $.  The lemma follows by
applying \autoref{theorem:sizecoverwithballs}
to bound the sizes of the sets $E_{w,\ell}$.
\end{proof}

We are now ready to state the main result of this subsection.
\begin{theorem} \label{theorem:kmediancoresetfactor}
Let $\inst = (P,k)$ be a $k$-median instance and
let $E_w$ be the set returned by the above MapReduce algorithm
for a fixed $\epsilon \in (0,1)$. Let
$\mathcal{A}$ be an $\alpha$-approximation algorithm for the
$k$-median problem, with constant $\alpha$. If $S$ is the solution returned by $\mathcal{A}$
with input $\inst' = (E_w,k)$, then $\nu_{P}(S)/\nu_{P}(\opt_\inst)
\leq \alpha + O(\epsilon)$.
\end{theorem} 
\begin{proof}
Let $\tau$ be the map from $P$ to $E_w$ of \autoref{definition:bounded}. By triangle inequality, it results that $\nu_{P}(S) \leq \sum_{x \in P}d(x,\tau(x)) + \nu_{E_w}(S)$.
The set $E_w$ is a $2\epsilon$-bounded coreset of $\inst$, so we have that $\sum_{x \in P}d(x,\tau(x)) \leq 2\epsilon \cdot \nu_{P}(\opt_\inst)$. Since $\mathcal{A}$ is an $\alpha$-approximation algorithm, we have that $\nu_{E_w}(S) \leq \alpha \cdot \nu_{E_w}(\opt_{\inst'})$. As $E_w$ is also a $7\epsilon$-centroid set, there exists a solution $X \subseteq E_w$ such that $\nu_{P}(X) \leq (1+7\epsilon)\nu_{P}(\opt_\inst)$. We obtain that $\nu_{E_w}(\opt_{\inst'}) \leq \nu_{E_w}(X) \leq (1+2\epsilon)(1+7\epsilon)\nu_{P}(\opt_\inst)$. In the last inequality, we used the fact that $E_w$ is a $2\epsilon$-approximate coreset of $\inst$ due to \autoref{lemma:boundedtostrongkmedian}. To wrap it up, $\nu_{P}(X)/\nu_{P}(\opt_\inst) \leq \alpha (1+7\epsilon)(1+2\epsilon) + 2\epsilon = \alpha + O(\epsilon)$.
\end{proof}

\subsection{Coreset construction for $k$-means}
\label{subsection:coresetkmeans}

In this subsection, we present a $2$-round MapReduce algorithm to
compute a weighted subset $E_w$ which is both an
$O(\epsilon^2)$-approximate coreset and a $O(\epsilon)$-centroid set
of an instance $\inst$ of $k$-means and then show that an
$\alpha$-approximate solution of $\inst' = (E_w,k)$ is an $(\alpha +
O(\epsilon))$-approximate solution of $\inst$.
The algorithm is an adaptation of the one devised in the previous
subsection for $k$-median, with suitable tailoring of the parameters
involved to account for the presence of squared distances in the
objective function of $k$-means. 

\noindent{\bf First Round.}
$P$ is partitioned into $L$ equally-sized subsets $P_1,\ldots,P_L$. Then in
parallel, on each $k$-means instance $\inst_\ell = (P_\ell,k)$, with
$\ell=1,\ldots,L$, the following steps are performed:
\begin{enumerate}
\item Compute a set $T_\ell$ of $m\geq k$ points such that $\mu_{P_\ell}(T_\ell) \leq \beta \cdot \mu_{P_\ell}(\opt_{\inst_\ell})$.
\item $R_\ell \longleftarrow  \sqrt{\mu_{P_\ell}(T_\ell)/|P_\ell|}$.
\item $C_{w,\ell} \longleftarrow \texttt{CoverWithBalls}(P_\ell,T_\ell,R_\ell,\sqrt{2}\epsilon,\sqrt{\beta})$.
\end{enumerate}
\noindent{\bf Second Round.} 
Let $C_w = \cup_{\ell =1}^{L}C_{w,\ell}$. The same partition of $P$ of the first round is used.  Together with $P_{\ell}$, the $\ell$-th reducer  receives a copy of  $C_w$, and all values $R_i$ computed in the previous round, for $i = 1, \ldots, L$. On each $k$-means instance $\inst_\ell = (P_\ell,k)$, with $\ell = 1,\ldots,L$, the following steps are performed:
\begin{enumerate} 
\item $R \longleftarrow \sqrt{\sum_{i=1}^{L} |P_i| \cdot R_i^2 / |P|}$
\item $E_{w,\ell} \longleftarrow \texttt{CoverWithBalls}(P_\ell,C_{w},R,\sqrt{2}\epsilon,\sqrt{\beta})$.
\end{enumerate}
The set $E_w = \cup_{\ell=1}^{L}E_{w,\ell}$ is the output of the
algorithm. The computation of $T_\ell$ in the first round can be accomplished using the  the linear-space 
constant approximation algorithms of \cite{Gupta2005,Kanungo2002}.

The analysis follows the lines of the one carried out for the $k$-median 
coreset construction. The following lemma establishes the properties of each $C_{w,\ell}$.
\begin{lemma}
\label{lemma:cwboundedcoresetkmeans}
For $\ell=1,\ldots,L$, $C_{w,\ell}$ is a $\epsilon^2$-bounded coreset 
of the $k$-means instance $\inst_\ell$.
\end{lemma}
\begin{proof}
Fix a value of $\ell$. Let $\tau_\ell$ be the map between the points in $C_{w,\ell}$ and the points in $P_\ell$ of \autoref{lemma:taucoverwithballs}. The set $C_{w,\ell}$ is weighted according to $\tau_\ell$. Also, it holds that:
\begin{align*}
    \sum_{x \in P_\ell}d(x,\tau_\ell(x))^2 \leq \frac{\epsilon^2}{2\beta} \sum_{x \in P_\ell}\left[R_\ell^2+d(x,T_\ell)^2 \right] \leq  \frac{\epsilon^2}{2\beta}\left[ R_\ell^2\cdot|P_\ell|+\mu_{P_\ell}(T_\ell) \right] \leq \epsilon^2 \cdot \mu_{P_\ell}(\opt_{\inst_\ell})
\end{align*}
\end{proof}

Next, in the following two lemmas, we characterize the properties and the
size of $E_w$.

\begin{lemma} 
\label{lemma:centroidkmeans}
Let $\inst = (P,k)$ be a $k$-means instance and assume that $\epsilon$
is a positive value such that $\epsilon+\epsilon^2 \leq 1/8$.  Then, the set
$E_w$ returned by the above MapReduce algorithm is both a
$4\epsilon^2$-bounded coreset and a $27\epsilon$-centroid set of
$\inst$.
\end{lemma}
\begin{proof}
 Let $\phi_\ell$ be the map of
  \autoref{lemma:taucoverwithballs} from the points in $P_\ell$ to the
  points in $E_{w,\ell}$. We have that $\sum_{x \in P_\ell}d(x,
  \phi_\ell(x))^2 \leq \epsilon^2/(2\beta)\left(|P_\ell| \cdot R_\ell^2
    + \mu_{P_\ell}(C_{w})\right)$. For any $x \in P$, let $\hat{\ell}$ be the index for which $x \in P_{\hat{\ell}}$, we define $\phi(x) = \phi_{\hat{\ell}}(x)$. We have that: 
    \begin{align*}
    \sum_{x \in P}d(x,\phi(x))^2 \leq \frac{\epsilon^2}{2\beta} \sum_{\ell = 1}^{L} \left[R^2 |P_\ell| + \mu_{P_\ell}(C_{w}) \right] = \frac{\epsilon^2}{2\beta}\left(\left( \sum_{\ell=1}^{L} |P_\ell|\cdot R_{\ell}^2\right) + \mu_{P}(C_w) \right)
    \end{align*}
  Using the fact that
  $|P_\ell|\cdot R_\ell^2 = \mu_{P_\ell}(T_\ell) \leq \beta \cdot
  \mu_{P_\ell}(\opt_{\inst_\ell}) \leq 4\beta \cdot
  \mu_{P_\ell}(\opt_{\inst})$, where the last inequality is due to \autoref{lemma:optimalsolrelation}, we have that $\sum_{\ell} R_\ell^2 |P_\ell| \leq \sum_{\ell} 4\beta \cdot
  \mu_{P_\ell}(\opt_\inst) \leq 4\beta \cdot \mu_{P}(\opt_\inst)$. Also, by \autoref{lemma:cwboundedcoresetkmeans} and \autoref{lemma:unionbounded}, $C_{w}$ is an
  $4\epsilon^2$-bounded coreset of $P$, thus
  $\mu_{P}(C_{w}) \leq 4\epsilon^2 \cdot
  \mu_{P}(\opt_{\inst})$.
  Therefore,  $E_{w}$ is
  an $4\epsilon^2$-bounded coreset of $\inst$.
  
  We now show that $E_w$ is a centroid set of $\inst$. Let $X = \{
  x^{E_w} : x \in \opt_\inst \}$. By \autoref{lemma:boundedtostrongkmeans}, $C_w$ is a
  $\gamma$-approximate coreset of $\inst$, with $\gamma = 4(\epsilon +
  \epsilon^2) \leq 1/2$. Hence, $\mu_{P}(X) \leq 1/(1-\gamma)\cdot
  \mu_{C_w}(X)$. By \autoref{proposition:squareddistance}, we have:
  \begin{align*}
      \mu_{C_w}(X) = \sum_{x \in C_w}w(x)d(x,X)^2 \leq (1+\epsilon)\mu_{C_w}(\opt_\inst) + (1+1/\epsilon)\sum_{x \in C_w}w(x)d(x^{\opt_\inst},X)^2
  \end{align*}
  Since $C_w$ is a $\gamma$-approximate coreset, it holds that
  $\mu_{C_w}(\opt_\inst) \leq (1+\gamma)\mu_{P}(\opt_\inst)$. By
  reasoning 
  as in the proof of \autoref{lemma:centroidkmedian}, we have that
  $\sum_{x \in C_w}w(x)d(x^{\opt_\inst},X)^2 \leq
  (5\epsilon^2/2 + \gamma\epsilon^2/2)\mu_{P}(\opt_\inst)$. Putting it
  all together, we conclude:
  \begin{align*}
  \mu_{P}(X)/\mu_{P}(\opt_\inst) \leq \left(1+\gamma+5\epsilon^2/2 + \gamma\epsilon^2/2 + 7\epsilon/2+3\gamma\epsilon/2\right)/(1-\gamma). 
  \end{align*}
  Since $\gamma \leq 1/2$, we have that $1/(1-\gamma) \leq 1 + 2\gamma$. By using the constraint on $\epsilon$ and the definition of $\gamma$, after some tedious computations, we obtain $\mu_{P}(X)/\mu_{P}(\opt_\inst) \leq 1+27\epsilon$.
  \end{proof}

\begin{lemma}
\label{lemma:kmeanssize}
Let $\inst = (P,k)$ be a $k$-means instance. Suppose that the points
in $P$ belong to a metric space with doubling dimension $D$. Let $E_w$
be the set returned by the above MapReduce algorithm with input
$\inst$ and $m\geq k$. Then, $|E_w| = O\left( L^2 \cdot m \cdot
(8\sqrt{2\beta}/\epsilon)^{2D} \log^2{|P|} \right)$
\end{lemma}

\begin{proof}
For any $\ell=1,\ldots,L$ and $x \in P_\ell$, it holds
that $R_{\ell}\cdot \sqrt{|P_\ell|} = \sqrt{\mu_{P_\ell}(T_\ell)} \geq
d(x,T_\ell)$. By using \autoref{theorem:sizecoverwithballs}, we obtain
that $|C_{w,\ell}| = O\left(m \cdot (8\sqrt{2\beta}/\epsilon)^{D}
\log{|P|} \right)$, and we can bound the size of $C_w$ with an union bound. By \autoref{lemma:cwboundedcoresetkmeans}, $C_{w,\ell}$ is a $\epsilon^2$-bounded coreset of $\inst_\ell$, hence $\mu_{P_\ell}(C_{w,\ell}) \leq \epsilon^2 \mu_{P_\ell}(\opt_{\inst_\ell})$. For any $x \in P$ we have that
$\epsilon\sqrt{|P|}\cdot R = \sqrt{\epsilon^2 \sum_{\ell} |P_\ell| R_\ell^2 } = \sqrt{\epsilon^2 \sum_{\ell}
  \mu_{P_\ell}(T_\ell)} \geq \sqrt{\epsilon^2 \sum_{\ell}
  \mu_{P_\ell}(\opt_{\inst_\ell})} \geq 
\sqrt{\sum_\ell \mu_{P_\ell}(C_{w,\ell})} \geq \sqrt{\mu_{P}(C_w)}  \geq d(x,C_{w})$. Thus, 
the lemma follows by applying \autoref{theorem:sizecoverwithballs} to
bound the sizes of the sets $E_{w,\ell}$.
\end{proof}

We are now ready to state the main result of this subsection.
\begin{theorem}
\label{theorem:kmeanscoresetfactor}
Let $\inst = (P,k)$ be a $k$-means instance and let $E_w$ be the set
returned by the above MapReduce algorithm for a fixed positive
$\epsilon$ such that $\epsilon + \epsilon^2 \leq 1/8$.  Let
$\mathcal{A}$ be an $\alpha$-approximation algorithm for the $k$-means
problem, with constant $\alpha$. If $S$ is the solution returned by $\mathcal{A}$ with input
$\inst' = (E_w,k)$, then $\mu_{P}(S)/\mu_{P}(\opt_\inst) \leq \alpha +
O(\epsilon)$.
\end{theorem} 
\begin{proof}
By \autoref{lemma:centroidkmeans} and \autoref{lemma:boundedtostrongkmeans}, $E_w$ is a $(4\epsilon^2+4\epsilon)$-approximate coreset of $\inst$. Therefore, $\mu_{P}(S) \leq (1/(1-4\epsilon-4\epsilon^2)) \cdot \mu_{E_w}(S)$. Since $\mathcal{A}$ is an $\alpha$-approximation algorithm, $\mu_{E_w}(S) \leq \alpha \cdot \mu_{E_w}(\opt_{\inst'})$. Also, $E_w$ is a $27\epsilon$-centroid set, thus there exists a solution $X \subseteq E_w$ such that $\mu_{P}(X) \leq (1+27\epsilon)\cdot \mu_{P}(\opt_\inst)$. We have that $\mu_{E_w}(\opt_{\inst'}) \leq \mu_{E_w}(X) \leq (1+4\epsilon+4\epsilon^2) \cdot \mu_P(X) \leq (1+4\epsilon+4\epsilon^2)(1+27\epsilon)\cdot \mu_P(\opt_\inst)$, where the second 
inequality follows again
from the fact that $E_w$ is a $(4\epsilon^2+4\epsilon)$-approximate coreset of $\inst$. Because of the constraints on $\epsilon$, we have that $1/(1-4\epsilon-4\epsilon^2) \leq 1+8\epsilon+8\epsilon^2$. Therefore, it finally results that $\mu_{P}(S)/\mu_{P}(\opt_\inst) \leq \alpha \cdot (1+8\epsilon+8\epsilon^2)(1+4\epsilon+4\epsilon^2)(1+27\epsilon) = \alpha + O(\epsilon)$.
\end{proof}

As noted in Subsection~\ref{subsection:approachkmedian}, a simpler version of this algorithm can be employed if we restrict our attention to the continuous case. Indeed, if we limit the algorithm to the first round and output the set $C_w = \cup_{\ell}C_{w,\ell}$, it is easy to show that an $\alpha$-approximate algorithm executed on the coreset $C_w$ returns a $(\alpha+O(\epsilon))$-approximate solution.

\subsection{MapReduce algorithms for $k$-median and $k$-means}
\label{subsection:mapreducefinal}

Let $\inst = (P,k)$ be a $k$-median (resp.,
$k$-means) instance. We can compute an approximate solution of $\inst$
in three MapReduce rounds: in the first two rounds, a weighted coreset $E_w$
is computed using the algorithm described in
Subsection~\ref{subsection:coresetkmedian} (resp.,
Subsection~\ref{subsection:coresetkmeans}), while in the third round
the final solution is computed by running a sequential approximation
algorithm for the weighted variant of the problem on $E_w$. 
Suppose that in the first of the two rounds of coreset construction we use
a linear-space algorithm to compute the sets $T_\ell$ 
of size $m = O(k)$,
and cost at most a factor $\beta$ times the optimal cost,
and that in the third round we run a linear-space $\alpha$-approximation
algorithm on $E_w$, with constant $\alpha$. Setting $L = \sqrt[\leftroot{-2}\uproot{2}3]{|P|/k}$ we obtain the following theorem
as an immediate consequence of 
Lemmas~\ref{lemma:kmediansize} and~\ref{lemma:kmeanssize}, 
and Theorems~\ref{theorem:kmediancoresetfactor} and~\ref{theorem:kmeanscoresetfactor}.
\begin{theorem}
Let $\inst = (P,k)$ be an instance of $k$-median $($resp.,
$k$-means$)$.  Suppose that the points in $P$ belong to a metric space
with doubling dimension $D$. For any $\epsilon \in (0,1)$
$($with $\epsilon+\epsilon^2 \leq 1/8$ for $k$-means$)$ 
the 3-round MapReduce algorithm
described above computes an $(\alpha+O(\epsilon))$-approximate
solution of $\inst$ using local space
$O\left(\ |P|^{2/3}k^{1/3} (16\beta/\epsilon)^{2D} \log^2{|P|}  \right)$ 
$($resp., $O\left(|P|^{2/3}k^{1/3} (8\sqrt{2\beta}/\epsilon)^{2D} \log^2{|P|}  \right)$$)$.
\end{theorem}

Note that for a wide range of the relevant parameters, 
the local space of the MapReduce algorithms is
substantially sublinear in the input size, and it is easy to show that the aggregate space is linear in $|P|$. 
As concrete instantiations of the above result, 
both the $T_{\ell}$'s and the final solution may be obtained through
the sequential algorithms in \cite{AryaGKMMP04} for
$k$-median, and in \cite{Gupta2005} for $k$-means. Both algorithms are
based on local search and feature approximations $\alpha = 3+2/t$
for $k$-median, and $\alpha = 5+4/t$ for $k$-means, where $t$ is the
number of simultaneous swaps allowed. With this choice, the result of
the above theorem holds with $\beta = \alpha = O(1)$. Alternatively,
for the  $T_{\ell}$'s we could use $k$-means++
\cite{BahmaniMVKV12}
as a bi-criteria approximation algorithm (e.g, see \cite{Wei16}),
which yields a smaller $\beta$, at the expense of a slight, yet
constant, increase in the size $m$ of the $T_\ell$'s. For larger $D$,
this might be a better choice as the coreset size (hence the local
memory) is linear in $m$ and $\beta^{2D}$ (resp., $\beta^D$).  Moreover, bi-criteria
approximations are usually faster to compute than actual solutions.

\section{Conclusions} \label{sec:conclusions}
We presented distributed coreset constructions that can
be used in conjunction with sequential approximation algorithms for
$k$-median and $k$-means in general metric spaces to obtain the first
space-efficient, 3-round MapReduce algorithms for the two problems,
which are almost as accurate as their sequential counterparts. The
constructions for the two problems are based on a uniform strategy,
and crucially leverage the properties of spaces of bounded doubling
dimension, specifically those related to ball coverings of sets of
points. One attractive feature of our constructions is their
simplicity, which makes them amenable to fast practical
implementations.

\bibliographystyle{plainurl}
\bibliography{references}

\end{document}